\documentclass[11pt]{article}
\usepackage{amsfonts}
\usepackage{graphicx}
\usepackage{amssymb,amsmath,epsfig}
\usepackage{amsthm}
\usepackage{color}
\include{epsf}

\newcommand{\be}{\begin{equation}}
\newcommand{\ee}{\end{equation}}
\def\beqa{\begin{eqnarray}}
\def\eeqa{\end{eqnarray}}
\def\nn{\nonumber}

\newtheorem{theorem}{Theorem}[section]

\newtheorem{lemma}[theorem]{Proposition}

\renewenvironment{thebibliography}[1]
         {\section*{References}\frenchspacing\small
          \begin{list}{[\arabic{enumi}]}
         {\usecounter{enumi}\parsep=2pt\topsep 0pt
         \settowidth{\labelwidth}{[#1]}
         \leftmargin=\labelwidth\advance\leftmargin\labelsep
         \rightmargin=0pt\itemsep=1pt\sloppy}}{\end{list}}
 \numberwithin{equation}{section}

\begin{document}

\title{$\kappa$-Minkowski star product in any dimension from symplectic realization}
\author{Anna Pacho{\l}$^{a}$ and Patrizia Vitale$^{b,c}$, }
\date{}
\maketitle
\begin{center}
\textit{$^{a}$Dipartimento di Matematica "Giuseppe Peano", Universit\`a degli
Studi di Torino,\\ Via Carlo Alberto, 10 - 10123 Torino, Italy}\\[0pt]
\textit{$^{b}$Dipartimento di Fisica Universit\`{a} di Napoli Federico II, \\ Via Cintia 80126 Napoli, Italy}\\ %
[0pt]
\textit{$^{c}$INFN, Sezione di Napoli, Via Cintia 80126 Napoli, Italy}\\[0pt]
e-mail: \texttt{apachol@unito.it, patrizia.vitale@na.infn.it}\\[1ex]
\end{center}

\begin{abstract}
We derive an explicit expression for the star product reproducing the $\kappa
$-Minkowski Lie algebra in any dimension $n$. The result is obtained by
suitably reducing the Wick-Voros star product defined on $\mathbb{C}%
^{d}_\theta$ with $n=d+1$. It is thus shown that the new star product can be
obtained from a Jordanian twist.

\end{abstract}

\vspace*{-1cm}

\section{Introduction}

Noncommutative space-time models represent an intermediate step in
understanding quantum aspects of space-time in the search for a quantum
theory of gravity. In analogy with quantum physics where quantum phase space
ceases to be a pseudo-Riemannian manifold and classical observables
(commuting functions defined on phase space) are replaced by operators, one
expects that in general relativity space-time as a dynamical variable
itself, becomes quantum at the Planck scale and space-time observables are
replaced by noncommuting operators. In particular space-time coordinate
functions are no longer classical variables but they belong to a
noncommutative `coordinate algebra'.

One of the first models of quantum space-time was proposed in the 1990's 
\cite{DFR}, together with the fuzzy sphere \cite{madore, fuzzy} (for review see \cite{fuzzyrev}), but already
in 1986 space-time noncommutativity was found within string theory \cite%
{witten}. In the context of algebraic geometry, noncommutative geometry has
been independently developed by Connes and collaborators \cite{Connes1}, 
\cite{Connes2}.

Here we shall follow the star product approach, where the noncommutative
algebra generated by coordinate functions is represented in terms of the
original commutative algebra of functions, with the pointwise multiplication
replaced by a noncommutative (star), associative product. There are many
known procedures for introducing star products, for example using
deformation quantization \cite{Konts}, which generalizes the canonical
quantization approach of classical phase space, or, in connection with
quantum groups, via twist deformation \cite{drinfeld, Tak}. Another
possibility, which is the one pursued in the present article, is the
symplectic realization proposed in \cite{selene}, where noncommutative
algebras are obtained as subalgebras of some ``canonical" algebra of Moyal
type.

The star product formulation of space-time noncommutativity is especially
relevant in physical applications, such as gravity \cite{Wess, Aschieri},
field and gauge theories (see e.g \cite{ncft} and refs. therein). Besides
constant noncommutativity, which has been widely explored in the literature,
there is a whole family of Lie algebra type noncommutativity which has
interesting properties with respect to symmetries. For example in three
dimensions, there has recently been a renewed interest in noncommutative
structures on $\mathbb{R}^3$ of Lie algebra type, in connection with their
occurrence in three-dimensional quantum gravity models \cite{QGrav}, where $%
\mathbb{R}^3$ is identified with the dual algebra of the local relativity
group. In such a framework star products are mainly introduced through a
group Fourier transform (see for example \cite{OR13} for a review and
comparison with other techniques), imposing compatibility with the group
convolution.

We shall focus here on the $\kappa $-Minkowski space-time which was
initially introduced in connection with the quantum (deformed) version
of the Poincar\'{e} symmetry (the so called $\kappa $-Poincar\'{e} group)
using the theory of quantum groups \cite{Luk1}. The $\kappa $-Minkowski
space-time has been first introduced in Refs.\cite{MR94,Z94}. It has been
further investigated by many authors \cite{KLM00}-\cite{Arzano09}, \cite{Agostini:2002de} in the
framework of noncommutative quantum field theory and Planck scale physics.

As for the derivation of a star product associated to the $\kappa $%
-Minkowski space-time, many proposals are available in the literature. Among
the others we quote \cite{DS11} where an integral form for the star product
is found for the three-dimensional case upon generalizing the deformation
quantization derivation of the Moyal product; in \cite{MSSG08} different
expressions are found based on the correspondence between realizations and
ordering prescriptions, whereas in \cite{Bu,BP1} the twist approach is
pursued.

In this paper we shall derive yet another expression for the $\kappa $%
-Minkowski star product, which agrees with all the others, as it should, at
the level of the star commutator of coordinate functions, but differs from
the others already at the first order in the noncommutative parameter for
the product of two generic functions. We start from the conjecture that most
noncommutative spaces can be obtained by reduction of some "big enough"
noncommutative space with canonical noncommutativity\footnote{%
Let us notice however that symplectic realizations of noncommutative
algebras do not require necessarily that the canonical noncommutative
algebra to be larger than the one we wish to realize. Examples can be found
in \cite{JS} where 3-d noncommutative spaces are realized in terms of the
2-d Moyal algebra.}. By the word canonical we mean the quantization of the
canonical symplectic form, which yields constant noncommutativity of the
Moyal type. Let us recall here that the Moyal product \cite{Moy} is just a
representative of a whole equivalence class of translation invariant star
products, all implementing the canonical star commutator among real
coordinate functions on $\mathbb{R}^{2d}$ 
\begin{equation}
y_{i}\star y_{j}-y_{j}\star y_{i}=i\theta _{ij},\,\,\,i,j=1,..,2d
\end{equation}%
and differing among them by ordering prescriptions. In the present paper we
shall use indeed the Wick-Voros product \cite{Berezin,Voros} which is the
one associated to normal ordering. For the case at hand, we will show that
by considering $d$ copies of the noncommutative plane equipped with the
Wick-Voros product we can obtain the $d+1\;\kappa $-Minkowski space,
together with a star product, as a subalgebra of the large starting algebra.
Additionally, we show that the star-product obtained in this way can be
obtained by a Jordanian twist, already known in the literature in a slightly
modified form.

The paper is organized as follows. In section \ref{starprod} we realize the $%
\kappa$-Minkowski algebra $\{x^\mu, \mu=0, ..., n\}$ in terms of complex
coordinate functions $\bar z^i, z^i, i= 1,..., d$ of the Wick-Voros plane $%
\mathbb{C}^d_\theta$, and derive an explicit expression of the star product
for the subalgebra of functions $f(x^\mu)$. The content of this derivation
constitutes the object of Proposition \ref{lemma1} and represents the main
result of the paper. In section \ref{sec3} we show that our star product,
obtained by symplectic realization, is actually obtainable by a Jordanian
twist operator. This result is formalized in Proposition \ref{lemma2}. We
finally discuss the issue of the existence of an integration measure with
respect to which our star product would be cyclic and we conclude that no
such measure exists, a result that we argue might be true for all Jordanian
twists. We conclude with final remarks.

\section{An explicit formula for the $\protect\kappa$-Minkowski star product}

{\label{starprod}}

Our starting point is the observation that the $\kappa $-Minkowski Lie
algebra in $n=d+1$ dimensions 
\begin{eqnarray}  \label{kmink}
{[}\hat{x}^{0},\hat{x}^{i}{]} &=&\hat{x}^{i}\;\;\;i=1,...,d  \label{kmink1}
\\
{[}\hat{x}^{i},\hat{x}^{j}{]} &=&0\;\;\;i,j=1,...,d  \label{kmink2}
\end{eqnarray}%
may be realized in terms of $d$ families of creation and annihilation
operators $a^{i},a^{i\dag }$, as follows 
\begin{equation}
\hat{x}^{0}=\sum_{i=1}^{d}a^{i\dag }a^{i},\;\;\;\hat{x}^{i}=a^{i\dag }
\end{equation}%
This realization\footnote{Notice that for the moment all operators are dimensionless.} , which is far from being unique, is a generalization of the
well known Jordan-Schwinger representation of all three-dimensional algebras
in terms of two uncoupled harmonic oscillators. Note that the map, as it is written above works in any dimension and is applicable for arbitrary Lie algebra.

The observation may be extended to Poisson realizations. We consider $d$
copies of the complex plane $\mathbb{C}$ and the algebra of functions $%
\mathcal{F}(\mathbb{C}^{d})$ equipped with the canonical Poisson bracket 
\begin{equation}
\{z^{i},\bar{z}^{j}\}=i\delta ^{ij}
\end{equation}%
The subset of functions $f(x^{\mu }),\,\mu =0,...,d$, with 
\begin{equation}
x^{0}=\sum_{i}\bar{z}^{i}z^{i},\;\;\;\;\;x^{i}=\bar{z}^{i}\;\;\;\;i=1,...,d
\label{kminkfunc}
\end{equation}%
is a Poisson subalgebra of $\mathcal{F}(\mathbb{C}^{d})$, with induced
Poisson brackets 
\begin{eqnarray}
\{x^{0},x^{i}\} &=&ix^{i}  \label{classkmink1} \\
\{x^{i},x^{j}\} &=&0  \label{classkmink2}
\end{eqnarray}%
and 
\begin{equation}
\{f(x),g(x)\}=\frac{\partial f}{\partial x^{\mu }}\{x^{\mu },x^{\nu }\}\frac{%
\partial g}{\partial x^{\nu }}=\partial _{\mu }f(x)\partial _{\nu
}g(x)ic_{\rho }^{\mu \nu }x^{\rho }
\end{equation}%
and $c_{\rho }^{\mu \nu }$ are the structure constants of the $\kappa $%
-Minkowski algebra. We indicate such an algebra with the symbol $\mathcal{A}%
^{n}$. The coordinate functions $x^{\mu },\mu =0,...,n-1$ polynomially
generate $\mathcal{A}^{n}$, with Poisson brackets of the $\kappa $-Minkowski
type.

As a second step, the commutative algebra of functions $\mathcal{F}(\mathbb{C%
}^{d})$ is made into a noncommutative algebra replacing the pointwise
product with the Moyal product 
\begin{equation}
f\star _{M}g\,(z^{i},\bar{z}^{i})=f(z,\bar{z})\exp \bigl[\frac{\theta }{2}(%
\overleftarrow{\partial }_{z^{i}}\overrightarrow{\partial }_{\bar{z}^{i}}-%
\overleftarrow{\partial }_{\bar{z}^{i}}\overrightarrow{\partial }_{z^{i}})%
\bigr]g(z,\bar{z}),\,\,\,\,i=1,...,d  \label{Moyal}
\end{equation}%
with $\theta $ a constant, real parameter. Alternatively, the Wick-Voros
product \cite{Berezin, Voros} can be used 
\begin{equation}
f\star _{WV}g\,(z^{i},\bar{z}^{i})=f(z,\bar{z})\exp \bigl[\theta 
\overleftarrow{\partial }_{z^{i}}\overrightarrow{\partial }_{\bar{z}^{i}}%
\bigr]g(z,\bar{z}),\,\,\,\,i=1,...,d  \label{Voros}
\end{equation}%
As well known, they both belong to the same family of translation invariant
star-products and are equivalent as they yield the same star-commutators for
the coordinate functions 
\begin{equation}
\lbrack z^{i},\bar{z}^{j}]_{\star }=\theta \delta ^{ij}  \label{theta_def}
\end{equation}%
only differing by symmetric terms.

In \cite{selene} the Moyal star product on $\mathbb{R}^{4}$ has been used to
induce, by different realizations, many nonequivalent star products on $%
\mathcal{F}(\mathbb{R}^{3})$, using the fact that all three dimensional
algebras could be realized as Poisson subalgebras of $\mathcal{F}(\mathbb{R}%
^{4})$, in the same spirit as above, and observing that those subalgebras
are also Moyal subalgebras. It was already noticed that complex realizations
are easier to deal with, by means of the identification $\mathbb{R}%
^{4}\simeq \mathbb{C}^{2}$. It was however difficult to find a closed
explicit form for such star products (only $x^{i}\star f(x)$ could be
explicitly calculated), whereas it was observed in \cite{hammou} (also see 
\cite{VW12} for applications) that the Wick-Voros product (\ref{Voros}) is
easier to handle and closed expressions for star products on $\mathbb{R}^{3}$
could be obtained by reduction. An example is the noncommutative algebra $%
\mathbb{R}_{\lambda }^{3}$, with the star product introduced in \cite{hammou}
which reproduces the $\mathfrak{su}(2)$ Lie algebra commutation relations as
a star-commutator of coordinate functions 
\begin{equation}
x^{i}\star x^{j}-x^{j}\star x^{i}=i\lambda \epsilon
_{k}^{ij}x^{k},\;\;i,j,k=1,..,3
\end{equation}%
with 
\begin{equation}
f\star g\,(x)=\exp \left[ \frac{\lambda }{2}\left( \delta
^{ij}x^{0}+i\epsilon _{k}^{ij}x^{k}\right) \frac{\partial }{\partial u^{i}}%
\frac{\partial }{\partial v^{j}}\right] f(u)g(v)|_{u=v=x}  \label{starsu2}
\end{equation}%
and $\lambda $ the noncommutativity parameter, to be identified with $\theta 
$ up to a constant. This was achieved on realizing the coordinates of $%
\mathbb{R}^{3}$ as quadratic-linear functions in $\bar{z}^{a},z^{a},\,a=1,2$ 
\begin{equation}
x^{\mu }=\frac{1}{2}\bar{z}^{a}\sigma _{ab}^{\mu }z^{b},\;\;\;\mu =0,..,3
\label{xmu}
\end{equation}%
with $\sigma ^{i}$ the Pauli matrices, $\sigma ^{0}$ the identity matrix $%
\mathbb{I}_{2}$, and observing that the subalgebra polynomially generated by
the coordinate functions $x^{\mu }$ and properly completed, is closed with
respect to the Wick-Voros star product\footnotemark \label{fnm:1} 
\footnotetext{%
Notice that $(x^{0})^{2}=\sum_{i}(x^{i})^{2}$. Therefore $x^{0}$ is
functionally dependent from the other coordinates. It represents the radius
in polar coordinates. This is not to be confused with the coordinate $x^{0}$
in the $\kappa $-Minkowski algebra which is an independent function\label%
{fnt:1}.}.

The same procedure can be applied to the $\kappa $-Minkowski algebra, Eqs. (%
\ref{kmink1}),(\ref{kmink2}), and extended to $d$ dimensions. The outcome is
contained in the following Proposition, which represents our main result

\begin{lemma}
\label{lemma1} Let us consider $d$ copies of the noncommutative plane
endowed with the Wick-Voros product, $({\mathcal{F}(\mathbb{C}^{d}),
\star_{WV}} )$ and the quadratic-linear functions defined in (\ref{kminkfunc}%
). Then, the Poisson subalgebra $\mathcal{A}^{d+1} \ni f(x^{0},x^{i}) $ is
also a noncommutative subalgebra of $(\mathcal{F}(\mathbb{\mathbb{C}}%
_{\theta }^{d}),\star _{WV})$ with induced star product 
\begin{equation}
(f\star g)(x)=\exp \Bigr[\theta (x^{0}\delta ^{\nu 0}+x^{i}\delta ^{\nu i})%
\frac{\partial }{\partial y^{0}}\frac{\partial }{\partial w^{\nu }}\Bigr]%
f(y)g(w)|_{y=w=x}  \label{kappa_stpr}
\end{equation}
\end{lemma}

\begin{proof}
This is obtained from Eq. (\ref{Voros}) observing that 
\begin{equation}
\frac{\overleftarrow{\partial }}{\partial z^{i}}\frac{\overrightarrow{%
\partial }}{\partial \bar{z}^{i}}=\bigl(\frac{\overleftarrow{\partial }}{%
\partial x^{0}}\bar{z}^{i}\bigr)\bigl(z^{i}\frac{\overrightarrow{\partial }}{%
\partial x^{0}}+\frac{\overrightarrow{\partial }}{\partial x^{i}}\bigr)=%
\frac{\overleftarrow{\partial }}{\partial x^{0}}x^{0}\frac{\overrightarrow{%
\partial }}{\partial x^{0}}+\frac{\overleftarrow{\partial }}{\partial x^{0}}%
x^{i}\frac{\overrightarrow{\partial }}{\partial x^{i}}.  \label{diffop}
\end{equation}%
In order to exponentiate this result and to prove (\ref{kappa_stpr}), the
crucial observation is that the two summands in  the RHS of (\ref{diffop}) commute. At
second order we have for instance 
\begin{eqnarray}
\Bigl(\frac{\overleftarrow{\partial }}{\partial z^{i}}\frac{\overleftarrow{%
\partial }}{\partial z^{j}}\Bigr)\Bigl(\frac{\overrightarrow{\partial }}{%
\partial \bar{z}^{i}}\frac{\overrightarrow{\partial }}{\partial \bar{z}^{j}}%
\Bigr) &=&\Bigl(\frac{\overleftarrow{\partial }^{2}}{(\partial x^{0})^{2}}\,%
\bar{z}^{i}\bar{z}^{j}\Bigr)\Bigl(z^{i}z^{j}\frac{\overrightarrow{\partial }%
^{2}}{(\partial x^{0})^{2}}+\frac{\overrightarrow{\partial }^{2}}{\partial
x^{i}\partial x^{j}}+z^{i}\frac{\overrightarrow{\partial }}{\partial x^{0}}%
\frac{\overrightarrow{\partial }}{\partial x^{j}}+z^{j}\frac{\overrightarrow{%
\partial }}{\partial x^{0}}\frac{\overrightarrow{\partial }}{\partial x^{i}}%
\Bigr)  \notag \\
&=&\frac{\overleftarrow{\partial }^{2}}{(\partial x^{0})^{2}}(x^{0})^{2}\frac{%
\overrightarrow{\partial }^{2}}{(\partial x^{0})^{2}}+\frac{\overleftarrow{%
\partial }^{2}}{(\partial x^{0})^{2}}x^{i}x^{j}\frac{\overrightarrow{\partial }%
^{2}}{\partial x^{i}\partial x^{j}}+2\frac{\overleftarrow{\partial }^{2}}{%
(\partial x^{0})^{2}}x^{0}x^{i}\frac{\overrightarrow{\partial }^{2}}{\partial
x^{i}\partial x^{0}}\nn
\end{eqnarray}%
where we have used 
\beqa
\frac{\overrightarrow{\partial }}{\partial \bar{z}^{i}}\frac{\overrightarrow{%
\partial }}{\partial \bar{z}^{j}}&=&\frac{\overrightarrow{\partial }}{\partial 
\bar{z}^{i}}(z^{j}\frac{\overrightarrow{\partial }}{\partial x^{0}}+\frac{%
\overrightarrow{\partial }}{\partial x^{j}})=z^{j}\frac{\overrightarrow{%
\partial }}{\partial \bar{z}^{i}}\frac{\overrightarrow{\partial }}{\partial
x^{0}}+\frac{\overrightarrow{\partial }}{\partial \bar{z}^{i}}\frac{%
\overrightarrow{\partial }}{\partial x^{j}}\nn\\
&=&z^{j}\bigl(z^{i}\frac{%
\overrightarrow{\partial }}{\partial x^{0}}+\frac{\overrightarrow{\partial }%
}{\partial x^{i}}\bigr)\frac{\overrightarrow{\partial }}{\partial x^{0}}+%
\bigl(z^{i}\frac{\overrightarrow{\partial }}{\partial x^{0}}+\frac{%
\overrightarrow{\partial }}{\partial x^{i}}\bigr)\frac{\overrightarrow{%
\partial }}{\partial x^{j}} \label{diffop2}
\eeqa
and this can be generalized to higher powers.
\end{proof}
As a corollary, the star product (\ref{kappa_stpr}) is associative, since
the Wick-Voros product on $\mathbb{C}^{d}_\theta$ is associative by
construction, being the trace over coherent states of the operator product 
\cite{Berezin}.

Notice that, had we started form the Moyal product (\ref{Moyal}) we should
have considered powers of the sum of two differential operators, the first
order in $\theta $ being 
\begin{equation}
\theta \Bigl[\frac{\overleftarrow{\partial }}{\partial z^{i}}\frac{%
\overrightarrow{\partial }}{\partial \bar{z}^{i}}-\frac{\overleftarrow{%
\partial }}{\partial \bar{z}^{i}}\frac{\overrightarrow{\partial }}{\partial
z^{i}}\Bigr]
\end{equation}%
which should be understood in terms of (\ref{diffop}), and its analogue with 
$z,\bar{z}$ exchanged, so to have 
\begin{eqnarray}
\theta \Bigl[\frac{\overleftarrow{\partial }}{\partial z^{i}}\frac{%
\overrightarrow{\partial }}{\partial \bar{z}^{i}}-\frac{\overleftarrow{%
\partial }}{\partial \bar{z}^{i}}\frac{\overrightarrow{\partial }}{\partial
z^{i}}\Bigr] &=&\theta \Bigl[\Bigl (\frac{\overleftarrow{\partial }}{%
\partial x^{0}}x^{0}\frac{\overrightarrow{\partial }}{\partial x^{0}}+\frac{%
\overleftarrow{\partial }}{\partial x^{0}}x^{i}\frac{\overrightarrow{%
\partial }}{\partial x^{i}}\Bigr)-\Bigl (\frac{\overleftarrow{\partial }}{%
\partial x^{0}}x^{0}\frac{\overrightarrow{\partial }}{\partial x^{0}}+\frac{%
\overleftarrow{\partial }}{\partial x^{i}}x^{i}\frac{\overrightarrow{%
\partial }}{\partial x^{0}}\Bigr)\Bigr]  \notag \\
&=&\theta \Bigl[\frac{\overleftarrow{\partial }}{\partial x^{0}}x^{i}\frac{%
\overrightarrow{\partial }}{\partial x^{i}}-\frac{\overleftarrow{\partial }}{%
\partial x^{i}}x^{i}\frac{\overrightarrow{\partial }}{\partial x^{0}}\Bigr]
\end{eqnarray}%
This time the two summands do not commute, which signals a problem in
expressing the star product in closed form. Indeed, considering higher
orders in $\theta$ the Moyal product expansion develops mixed terms. For
example at second order in $\theta $ we find mixed terms of the kind 
\begin{equation}
\frac{\overrightarrow{\partial }}{\partial z^{i}}\frac{\overrightarrow{%
\partial }}{\partial \bar{z}^{j}}=\frac{\overrightarrow{\partial }}{\partial
z^{i}}(z^{j}\frac{\overrightarrow{\partial }}{\partial x^{0}}+\frac{%
\overrightarrow{\partial }}{\partial x^{j}})
\end{equation}%
where, differently from (\ref{diffop2}), the first and second differential
operator do not commute\footnote{%
However, this is not a proof that a closed form for the Moyal based
star-product does not exist. We will show in the next section that the WV
based star product is related to a non-symmetric Jordanian twist. It would
be interesting to investigate whether a symmetrized Jordanian twist might
yield the Moyal based product.}. This problem was already noticed for the $%
\mathfrak{su}(2)$-like star product and led to the Wick-Voros based star
product (\ref{starsu2}) as opposed to the Moyal based one \cite{selene}.

In order to compare the $\kappa$-Minkowski star product with Eq. (\ref%
{starsu2}) we may rewrite it in terms of the structure constants of the $%
\kappa$-Minkowski algebra so to obtain 
\begin{equation}
f\star g(x)=\exp \Bigr[\theta (x^{0}\delta ^{\nu 0}+x^{i} c^{0\nu}_i)\frac{%
\partial }{\partial y^{0}}\frac{\partial }{\partial w^{\nu }}\Bigr]%
f(y)g(w)|_{y=w=x} .  \label{stpr2}
\end{equation}%
Despite the similarity of this expression with Eq. (\ref{starsu2}), the two
products are very different. The reason, already enunciated in the footnote $%
\;^{\ref{fnt:1}}$ is that here $x^0$ is an independent coordinate function
whereas in Eq. (\ref{starsu2}) it is functionally dependent from the other
ones: the latter generates the non-trivial center of the algebra (which is
related to the Casimir of $SU(2)$), while the $\kappa$-Minkowski algebra has
no non-trivial center (the corresponding group has no Casimir).

The result of Proposition \ref{lemma1} does not depend on space-time
dimension and has the advantage of being given in terms of a closed
expression unlike other known results in the literature.

For further developments in next sections it is useful to write the star
product (\ref{kappa_stpr}) (equiv. (\ref{stpr2})) as a series expansion. We
have 
\begin{equation}
\exp \Bigr[\theta (x^{0}\delta ^{\nu 0}+x^{i}\delta ^{\nu i})\frac{\partial 
}{\partial y^{0}}\frac{\partial }{\partial w^{\nu }}\Bigr]=\exp \Bigr[\theta 
\frac{\partial }{\partial y^{0}}x^{\nu }\frac{\partial }{\partial w^{\nu }}%
\Bigr]=\sum_{n=0}^{\infty }\frac{\theta ^{n}}{n!}\left( \frac{\partial }{%
\partial y^{0}}\right) ^{n}\left( x^{\nu }\frac{\partial }{\partial w^{\nu }}%
\right) ^{n}.  \label{sexp}
\end{equation}

We shall come back to this expression in section \ref{sec3}, where we will
show that our star product can be obtained as a twisted product with
Jordanian twist.

\section{Twist operators and deformed symmetries}

\label{sec3}

One can relate the star product formulation with the twist-deformation
approach, but only for a certain type of quantum deformations. The
deformations leading to noncommutative algebras (as noncommutative
space-times) are related to the Hopf algebras formalism. Let us then recall
a few important features of this approach.

Let us denote with $H$ a Hopf algebra, equipped with coproduct, counit and
antipode maps $\left( H,\Delta ,\epsilon ,S\right) $. Hopf algebras
generalize in noncommutative geometry the symmetries of space-time. The
classical example of a Hopf algebra is provided by $U\mathfrak{g}$, the
universal enveloping algebra of some Lie algebra $\mathfrak{g}$, with the
Hopf algebra maps $(\Delta ,\epsilon ,S)$ appropriately defined.

Let us stick to the latter situation, with $\mathfrak{g}$ the Lie algebra of
space-time symmetries. The Hopf algebra $U\mathfrak{g}$ can undergo a
deformation procedure and the deformation may be obtained trough a twist
operator, $F$ \cite{drinfeld,Tak}, 
which is an invertible element in the tensor product of Hopf algebras $%
H\otimes H$. One of the advantages of the twist deformation is that it
provides straightforwardly the universal quantum $R$ matrix and an explicit
formula for the star product on the Hopf module algebra, which is consistent
with Hopf-algebra actions related to $H$. Issues such as symmetries and
invariances can thus be properly addressed.

The twist fulfils the cocycle condition: 
\begin{equation}
\left( F\otimes 1\right) \left( \Delta \otimes 1\right) F=\left( 1\otimes
F\right) \left( 1\otimes \Delta \right) F  \label{cocycle}
\end{equation}%
and the normalization condition 
\begin{equation}
(1\otimes \epsilon ){F}=(\epsilon \otimes 1){F}=1\otimes 1.  \label{normal}
\end{equation}%
Under the twisted deformation the commutation relations of the Lie algebra
sector of the Hopf algebra do not change but the coalgebra structure is
suitably modified via relations: 
\begin{eqnarray}
\Delta _{F}(X) &=&{F}\Delta (X){F}^{-1}  \label{tw_cop} \\
\varepsilon (X) &=&0,\quad S_{F}(X)=\mathrm{f}^{\alpha }S(\mathrm{f}_{\alpha
})S(X)S(\bar{\mathrm{f}}^{\beta })\bar{\mathrm{f}}_{\beta }.  \label{tw_S}
\end{eqnarray}%
%
where we used the notation ${F}=\mathrm{f}^{\alpha }\otimes \mathrm{f}%
_{\alpha },\quad {F}^{-1}=\bar{\mathrm{f}}^{\alpha }\otimes \bar{\mathrm{f}}%
_{\alpha }$\footnote{%
For each value of $\alpha $, $\mathrm{f}^{\alpha }$ and $\mathrm{f}_{\alpha
} $ (and similarly $\bar{\mathrm{f}}^{\alpha }$ and $\bar{\mathrm{f}}%
_{\alpha } $) are different elements of $H$.}. Such a deformation of the
co-structures in the universal enveloping algebra of symmetries implies a
simultaneous deformation in the space-time algebra. In this way the
commutative multiplication in the algebra of functions\footnote{%
Twist deformations can be introduced for general differentiable manifolds $M$%
, however here we are only interested in $M=\mathbb{R}^{d+1}$ to be
consistent with the previous section.} $\mathcal{A}^{d+1}$ is replaced by a
new twisted one: 
\begin{equation}
f\star _{F}g=\mu \circ {F}^{-1}(f\otimes g)=\bar{\mathrm{f}}^{\alpha }(f)%
\bar{\mathrm{f}}_{\alpha }(g),\quad f,g\in \mathcal{A}^{d+1}  \label{stpr}
\end{equation}%
where $\mu $ is the usual point-wise multiplication. As a result one obtains
a noncommutative algebra $(\mathcal{A}^{d+1},\star _{F})$, which we identify
with the noncommutative space-time as before. In the case of triangular
deformations it is possible to associate with a given star product defined
through a quantization-dequantization scheme \footnote{For details on the quantization and dequantization operators see \cite{stratonovich}, \cite{quant-dequant} and references therein.} the corresponding twisting
element.

For example considering the Moyal product defined in (\ref{Moyal}) we can
easily find the corresponding twist as: 
\begin{equation}
F_M=\exp \left( -\frac{\theta }{2}(\partial _{z^{i}}\otimes \partial _{\bar{z%
}^{i}}-\partial _{\bar{z}^{i}}\otimes \partial _{z^{i}})\right)  \label{F_M}
\end{equation}%
Such a twist is called Abelian\footnote{%
Analogously we can define a twist operator for the Wick-Voros product \cite%
{GLV08}, which is the one used in this paper, as $F_{\theta }=\exp \left( -{%
\theta }(\partial _{z^{i}}\otimes \partial _{\bar{z}^{i}})\right) $, with
support in the same Abelian algebra. The two twists only differ by
co-boundary terms \cite{GLV09}.} because it has support in the Abelian (Lie)
algebra generated by $\mathfrak{a}=\mathrm{span}\{\partial _{z^{i}},\partial
_{\bar{z}^{i}}\},i=1...d$. It is easy to define a Hopf algebra structure on
its undeformed universal enveloping algebra $U{\mathfrak{a}}$ by defining
all the maps on the generating elements as follows:

-coproduct $\Delta \left( \partial _{z^{i}}\right) =\partial _{z^{i}}\otimes
1+1\otimes \partial _{z^{i}}\quad ;\quad \Delta (\partial _{\bar{z}%
^{i}})=\partial _{\bar{z}^{i}}\otimes 1+1\otimes \partial _{\bar{z}^{i}}$

-counit $\epsilon (\partial _{z^{i}})=0\quad ;\quad \epsilon (\partial _{%
\bar{z}^{i}})=0$

-antipode $S\left( \partial _{z^{i}}\right) =-\partial _{z^{i}}\quad ;\quad
S\left( \partial _{\bar{z}^{i}}\right) =-\partial _{\bar{z}^{i}}$

and then extending them to all the elements of $U{\mathfrak{a}}$.

Let us consider now the deformation of the Hopf algebra $U{\mathfrak{a}}$
associated with the Moyal twist (\ref{F_M}). We can immediately see that the
cocycle condition (\ref{cocycle}) is satisfied by $F_M$. To be more precise
the twisted deformation of any Lie algebra $\mathfrak{g}$ requires a
topological extension of the corresponding enveloping algebra $U\mathfrak{g}$
into an algebra of formal power series $U\mathfrak{g}[[\theta ]]$ in the
formal parameter $\theta $. We have $U_{F_M}{\mathfrak{a}}[[\theta]]\equiv
\left( U{\mathfrak{a}}[[\theta]],\Delta _{F_M},S_{F_M},\epsilon \right) $
with the twisted coproduct and the antipode maps obtained by (\ref{tw_cop}-%
\ref{tw_S}). We notice that the coalgebra part stays undeformed under the
action of this twist, i.e.: 
\begin{eqnarray}
\Delta _{F_M}\left( \partial _{z^{i}}\right) &=&\Delta _{0}\left( \partial
_{z^{i}}\right) \quad;\quad\Delta _{F_M}\left( \partial _{\bar{z}%
^{i}}\right) =\Delta _{0}\left( \partial _{\bar{z}^{i}}\right) \\
S_{F_M}\left( \partial _{z^{i}}\right) &=&-\partial _{z^{i}}\quad ;\quad
S_{F_M}\left( \partial _{\bar{z}^{i}}\right) =-\partial _{\bar{z}^{i}}.
\end{eqnarray}

The case we are interested in this paper, however, is related to a slightly
more complicated kind of space-time noncommutativity, which is the $\kappa $%
-Minkowski star product defined in Eq. (\ref{kappa_stpr}) or equivalently (%
\ref{sexp}). We prove the following

\begin{lemma}
\label{lemma2} The star product defined in Eq. (\ref{sexp}) is associated to
a twist operator 
\begin{equation}
f\star g=\mu \circ \left( F_{J}^{-1}\left( f\otimes g\right) \right)
\label{tobeproven}
\end{equation}%
with $F_{J}$ the Jordanian twist \cite{Jord1} 
\begin{equation}
F_{J}=\exp \left( \sigma \otimes J\right) \quad ;\quad \sigma =\ln \left(
1+\theta P_{0}\right) \quad  \label{twist}
\end{equation}%
The twist $F_{J}$ is an element of $\left( U\mathfrak{b}\otimes U\mathfrak{b}%
\right) \left[ \left[ \theta \right] \right] $ with $\mathfrak{b}=\mathrm{%
span}\{P_{0},J;\left[ J\,,P_{0}\right] =P_{0}\}$ the two-dimensional Borel
subalgebra of $\mathfrak{gl}(2,\mathbb{C})$.
\end{lemma}

\begin{proof}
Using the representation for the generators  $J=-x^{0}\partial
_{0}-x^{k}\partial _{k}=-x^{\mu }\partial _{\mu }\quad \,$and$\quad
P_{0}=\partial _{0}$ we can expand the inverse of the twist in the following
way:
\begin{equation}
F_{J}^{-1}=\exp \left( \ln \left( 1+\theta \partial _{0}\right) \otimes
x^{\mu }\partial _{\mu }\right) =\left( 1+\alpha \partial _{0}\right)
_{\left( 1\right) }^{\left( x^{\mu }\partial _{\mu }\right) _{\left(
2\right) }}=\sum_{n=0}^{\infty }\frac{1}{n!}\left( \theta \partial
_{0}\right) _{\left( 1\right) }^{n}\cdot \left( x^{\mu }\partial _{\mu
}\right) _{\left( 2\right) }^{\underline{n}}
\end{equation}
where $y^{\underline{n}}=y\left( y-1\right) ...\left( y-n+1\right) $ and $y^{%
\underline{1}}=y\quad ;\quad y^{\underline{0}}=1$. We can re-write the twist
as: 
\begin{equation}
F_{J}^{-1}=\sum_{n=0}^{\infty }\frac{\theta ^{n}}{n!}\left( \partial
_{0}\right) _{\left( 1\right) }^{n}\cdot \left( x^{\mu }\partial _{\mu
}\right) _{\left( 2\right) }^{\underline{n}}=\sum_{n=0}^{\infty }\frac{%
\theta ^{n}}{n!}\left( \partial _{y_{0}}\right) _{\left( 1\right) }^{n}\cdot
\left( x^{\mu }\partial _{w_{\mu }}\right) _{\left( 2\right) }^{n}
\end{equation}%
where 
\begin{equation*}
\left( x^{\mu }\partial _{x^{\mu }}\right) ^{\underline{n}}=x^{\alpha
}\partial _{x^{\alpha }}\left( x^{\beta }\partial _{x^{\beta }}-1\right)
\left( x^{\gamma }\partial _{x^{\gamma }}-2\right) \cdot ....\cdot \left(
x^{\mu }\partial _{x^{\mu }}-n+1\right)
\end{equation*}%
and the notation $x^{\mu }\partial _{w^{\mu }}$ is a way to indicate that
the differential operator $\partial _{w^{\mu }}$ doesn't act on functions of 
$x$ but only on functions of $w$. Indeed we notice that 
\begin{eqnarray}
\left( x^{\mu }\partial _{x^{\mu }}\right) ^{\underline{2}} &=&\left(
x^{\alpha }\partial _{x^{\alpha }}\right) x^{\beta }\partial _{x^{\beta
}}-x^{\alpha }\partial _{x^{\alpha }}=x^{\beta }\partial _{x^{\beta
}}+x^{\alpha }x^{\beta }\partial _{x^{\alpha }}\partial _{x^{\beta
}}=x^{\alpha }x^{\beta }\partial _{x^{\alpha }}\partial _{x^{\beta }}  \notag
\\
&=&\left( x^{\mu }\partial _{w_{\mu }}\right) ^{2}
\end{eqnarray}
 and by induction we can prove that 
\begin{equation}
\left( x^{\mu }\partial _{\mu }\right) ^{\underline{n}}=\underbrace{x^{\mu
}x^{\rho }...x^{\lambda }}_{\text{n}}\ \underbrace{\partial _{\mu }\partial
_{\rho }...\partial _{\lambda }}_{\text{n}}\ =\left( x^{\mu }\partial
_{w}^{\mu }\right) ^{n}
\end{equation}
Upon substituting this result into Eq. (\ref{tobeproven}) we obtain 
\begin{equation}
f\star g(x)=\Bigl[\mu \circ \Bigl(F_{J}^{-1}\left( f\otimes g\right) (w)%
\Bigr)\Bigr]_{w=x}=\sum_{n=0}^{\infty }\frac{\theta ^{n}}{n!}\left( \partial
_{y^{0}}\right) ^{n}f(y)\left( x^{\mu }\partial _{w_{\mu }}\right)
^{n}g(w)|_{y=w=x}
\end{equation}%
which agrees with Eq. (\ref{sexp}). \end{proof}

The twist $F_{J}$ (first introduced in \cite{Jord2}, in a slightly different
form) has support in the Borel algebra $\mathfrak{b}$. However, if we are
interested in the symmetries of the $\kappa $-Minkowski space-time, we use
the fact that the above two-dimensional Borel algebra $\mathfrak{b}$ is a
subalgebra of the Poincar\'{e}-Weyl algebra (one generator extension of the
Poincar\'{e} algebra) $\mathfrak{pw}=\mathrm{span}\{M_{\mu \nu },P_{\mu
},J\},\mu ,\nu =0,1,2,3$ with the following commutation relations\footnote{%
We are using anti-hermitean generators.}: 
\begin{eqnarray}
\left[ M_{\mu \nu },M_{\rho \lambda }\right] &=&\eta _{\nu \rho }M_{\mu
\lambda }+\eta _{\mu \lambda }M_{\nu \rho }-\eta _{\mu \rho }M_{\nu \lambda
}-\eta _{\nu \lambda }M_{\mu \rho } \\
\left[ M_{\mu \nu },P_{\rho }\right] &=&\eta _{\nu \rho }P_{\mu }-\eta _{\mu
\rho }P_{\nu } \\
\left[ P_{\mu },P_{\nu }\right] &=&0\quad ;\quad \left[ J,P_{\mu }\right]
=P_{\mu }\quad ;\quad \left[ J,M_{\mu \nu }\right] =0
\end{eqnarray}%
Again we can turn $U\mathfrak{pw}$ into a Hopf algebra by defining the
following maps on the generators: $\Delta \left( J\right) =J\otimes
1+\otimes J\quad ,\;\Delta (P_{\mu })=P_{\mu }\otimes 1+1\otimes P_{\mu
}\quad ,\;\Delta (M_{\mu \nu })=M_{\mu \nu }\otimes 1+1\otimes M_{\mu \nu },$

$\epsilon (J)=\epsilon (P_{\mu })=\epsilon (M_{\mu \nu })=0\quad ,\;S\left(
J\right) =-J\quad ,\;S\left( P_{\mu }\right) =-P_{\mu }\quad ,\;S\left(
M_{\mu \nu }\right) =-M_{\mu \nu }$ and then extending them to the whole $U%
\mathfrak{pw}$.

After twisting the coalgebra sector of $U\mathfrak{pw}$ with $F_{J}$ via
Eqs. (\ref{tw_cop}-\ref{tw_S}) we get its twisted version $U\mathfrak{pw}%
_{F_{J}}\left[ \left[ \theta \right] \right] $(cf.\cite{BP1}) with: 
\begin{eqnarray}
\Delta _{F_{J}}\left( P_{\mu }\right) &=&P_{\mu }\otimes 1+e^{\sigma
}\otimes P_{\mu } \\
\Delta _{F_{J}}\left( M_{\mu \nu }\right) &=&M_{\mu \nu }\otimes 1+1\otimes
M_{\mu \nu }+\theta \left( \eta _{\nu 0}P_{\mu }-\eta _{\mu 0}P_{\nu
}\right) e^{-\sigma }\otimes J \\
\Delta _{F_{J}}\left( J\right) &=&J\otimes 1+1\otimes J-\theta \partial
_{0}e^{-\sigma }\otimes J=J\otimes 1+e^{-\sigma }\otimes J \\
S_{F_{J}}\left( P_{\mu }\right) &=&-P_{\mu }e^{-\sigma }\quad ;\quad
S_{F_{J}}\left( M_{\mu \nu }\right) =-M_{\mu \nu }+\theta \left( \eta _{\nu
0}P_{\mu }-\eta _{\mu 0}P_{\nu }\right) J\quad ;\quad S_{F_{J}}\left(
J\right) =-e^{\sigma }J  \notag
\end{eqnarray}%
The Poincar\'{e}-Weyl Lie algebra $\mathfrak{pw}$ is a subalgebra of a
bigger one, the inhomogeneous general linear algebra $\mathfrak{igl}(n)$,
which was investigated in the context of non-symmetric Jordanian twist in 
\cite{BP1}.

We may associate to a given twist deformation a realization of noncommuting
coordinate functions in terms of differential operators, given by the
following relations:

\begin{itemize}
\item left-handed realization: 
\begin{equation}
\hat{x}_{L}^{\mu }=\left( \bar{\mathrm{f}}^{\alpha }\triangleright x^{\mu
}\right) \cdot \bar{\mathrm{f}}_{\alpha }
\end{equation}

\item right-handed realization: 
\begin{equation}
\hat{x}_{R}^{\mu }=\bar{\mathrm{f}}^{\alpha }\cdot \left( \bar{\mathrm{f}}%
_{\alpha }\triangleright x^{\mu }\right)
\end{equation}
\end{itemize}

In the case of the Jordanian twist (\ref{twist}) we obtain these
differential operators in the form: 
\begin{eqnarray}
\hat{x}_{L}^{\mu }&=&x^{\mu }+\theta \delta _{0}{}^{\mu }x^{\nu }\partial
_{\nu }  \notag \\
\hat{x}_{R}^{\mu }&=&x^{\mu }\left( 1+\theta \partial _{0}\right) .
\end{eqnarray}
Let us notice that such realizations are easily obtained from the explicit
expression of the star product, Eq. (\ref{sexp}), by star-multiplying a
generic function $f$ on the left (resp. on the right) with $x^\mu$.
Moreover, it is easy to check that using such realizations (left and right
respectively) the $\kappa $-Minkowski relations are satisifed: 
\begin{equation}
\left[ \hat{x}_{L,R}^{0},\hat{x}_{L,R}^{k}\right] =\pm \theta \hat{x}%
_{L,R}^{k}.
\end{equation}

\subsection{Integration measure}

One of the advantages of having a star product obtained from a twist is that
the differential calculus is completely determined by the twist operator 
\cite{Wess,Aschieri}, (also see \cite{ALV08} for generalization to field
theory). The guiding principle to obtain all the geometric structures
related to the noncommutative algebra of functions is the following. Every
time we have a bilinear map 
\begin{equation}
\mu :X\times Y\rightarrow Z
\end{equation}%
where $X,Y,Z$ are modules over the algebra $\mathcal{A}$, with an action of
the twist $F$ on $X$ and $Y$, we combine the map with the action of the
twist so to obtain a deformed map $\mu _{\star }$ 
\begin{equation}
\mu _{\star }=\mu \circ F^{-1}.  \notag
\end{equation}%
In the particular case of the wedge product of elementary one-forms deformed
via $F_{J}$ (\ref{twist}) we have 
\begin{eqnarray}
dx^{\mu }\wedge _{\star }dx^{\nu } &=&\wedge \;\circ \;F^{-1}(dx^{\mu
}\otimes dx^{\nu })={\bar{\mathrm{f}}}^{\alpha }(dx^{\mu })\wedge \bar{%
\mathrm{f}}_{\alpha }(dx^{\nu })  \notag \\
&=&dx^{\mu }\wedge dx^{\nu }+\theta \mathcal{L}_{\partial _{0}}(dx^{\mu
})\wedge \mathcal{L}_{x^{\sigma }\partial _{\sigma }}(dx^{\nu })+O(\theta
^{2})=dx^{\mu }\wedge dx^{\nu }  \label{wedge}
\end{eqnarray}%
with $\mathcal{L}_{X}$ the Lie derivative with respect to $X$, so that the
volume form on $\mathbb{R}^{d+1}$ is the undeformed one $\Omega
=dx^{0}\wedge dx^{1}\wedge ..\wedge dx^{d}$. The star product (\ref{sexp})
however, is not cyclic with respect to $|\Omega |$ 
\begin{equation}
\int {\mathrm{d}}^{d+1}x\;f\star g\neq \int {\mathrm{d}}^{d+1}x\;g\star f
\end{equation}%
the reason being that one of the generators yielding the twist operator is
the dilation, $J=-x^{\mu }\partial _{\mu }$, which does not preserve the
volume form $\Omega $. For (\ref{sexp}) to be cyclic, we could consider to
modify the integral measure by the introduction of a measure function $%
h\left( x\right) $ such that : 
\begin{equation}
\int h(x)d^{d+1}x\,f\star g=\int h(x)d^{d+1}x\,g\star f.
\end{equation}%
On expanding the star product in powers of $\theta $ by means of (\ref{sexp}%
) we obtain at first order in $\theta $ 
\begin{equation}
\left[ \int h(x)(f\star g-g\star f)\right] _{\theta }=\theta \int h(x){x^{k}}%
\left( {\partial _{0}f\cdot \partial _{k}g-\partial _{k}f\cdot \partial _{0}g%
}\right)  \label{firstord}
\end{equation}%
thus, upon integrating by parts, we obtain that the cyclicity condition is
satisfied at first order in $\theta $ if $h(x)$ satisfies the following
equations: 
\begin{equation}
\partial _{0}h\left( x\right) =0\quad ,\quad x^{k}\partial _{k}h\left(
x\right) =-d\,h\left( x\right)  \label{measure_eq}
\end{equation}%
Examples of solutions are \cite{Moller} in $d+1$-dimensions: 
\begin{eqnarray}
&&r^{-d}\;,\text{where}\;r=\sqrt{\sum_{i}x_{i}^{2}}  \notag  \label{condts}
\\
&&\left( \Pi _{i=1}^{d}x_{i}\right) ^{-1}  \notag \\
&&\left( \sum_{i=1}^{d}\left( x_{i}\right) ^{k}\right) ^{-\frac{d}{k}}
\end{eqnarray}

However, such choices break down already at second order. Indeed, after some
calculation we obtain 
\begin{eqnarray}
&&\left[ \int {\mathrm{d}}^{d+1}x\;h\left( x\right) f\star g-\int {\mathrm{d}%
}^{d+1}x\;h\left( x\right) g\star f\right] _{\theta ^{2}}  \notag \\
&=&\frac{\theta ^{2}}{2}\int {\mathrm{d}}^{d+1}xh\left( x\right) \{x^{\rho
}x^{\mu }\left( \partial _{0}^{2}f\cdot \partial _{\rho }\partial _{\mu
}g-\partial _{\rho }\partial _{\mu }f\cdot \partial _{0}^{2}g\right) \}
\label{2ndord}
\end{eqnarray}%
On integrating by parts and on using the first order conditions (\ref%
{measure_eq}), we get that the only solution is $h\left( x\right) =0$.

Another possibility would be to consider a natural modification of the star
product here introduced. Indeed, the star product in Proposition \ref{lemma1}
by reduction of the Wick-Voros product singles out just one of the possible
twists reproducing $\kappa $-Minkowski relations, i.e. the non-symmetric
Jordanian twist (\ref{twist}) \cite{Jord1,Jord2}. However there exists a
possible symmetrization of this twist, introduced in \cite%
{Tolstoy:arXiv:0712.3962} 
which we here briefly review.

For $F=\mathrm{f}^{\alpha }\otimes \mathrm{f}_{\alpha }$, a twisting
two-tensor of the Hopf algebra $\left( U{\mathfrak{g}},\Delta ,S,\epsilon
\right) $ and $\lambda =\sqrt{\mathrm{f}^{\alpha }S\left( \mathrm{f}_{\alpha
}\right) }$ there exists a related twisting two-tensor $F_{r}^{\left(
\lambda \right) }$: 
\begin{equation}
F_{r}^{\left( \lambda \right) }:=\lambda ^{-1}\otimes \lambda ^{-1}F\Delta
\left( \lambda \right)  \label{symm_tw}
\end{equation}%
which is locally $r$-symmetric, with $r$ the classical $r$ matrix. One says
that the twist is locally r-symmetric if it satisfies the cocycle (\ref%
{cocycle}) and normalization (\ref{normal}) conditions and if its expansion 
in powers of the deformation parameter $\theta $ has the form: $F_{r}\left(
\theta \right) =1+c\theta r+O\left( \theta ^{2}\right) $ where $c\neq 0$ is
a numerical coefficient.

One can notice immediately that the Jordanian twist (\ref{twist}) 
is not $r-$symmetric, since the classical $r-$matrix associated with it is $%
r=P_{0}\wedge J$.

The corresponding $r-$symmetric version of (\ref{twist}) is given by \cite%
{Tolstoy:arXiv:0712.3962} 
\begin{eqnarray}
F_{Jrs} &=&\exp (\frac{\theta }{2}(JP_{0}\otimes 1+1\otimes JP_{0})\,)\exp
(\ln \left( 1+\theta P_{0}\right) \otimes J)\times   \notag \\
&&\times \exp \left( -{\frac{\theta }{2}}\left( JP_{0}{\otimes 1+J\otimes
P_{0}+P_{0}\otimes J+1\otimes }JP_{0}\right) \right) 
\end{eqnarray}%
Its inverse 
can be re-written using the differential representation of the generators $J$
and $P_{0}$: 
\begin{eqnarray}
F_{Jrs}^{-1} &=&\,\exp \left( -{\frac{\theta }{2}}\left( x^{\mu }\partial
_{\mu }\partial _{0}{\otimes 1+x^{\mu }\partial _{\mu }\otimes \partial
_{0}+\partial _{0}\otimes x^{\mu }\partial _{\mu }+1\otimes }x^{\mu
}\partial _{\mu }\partial _{0}\right) \right) \exp (\ln \left( 1+\theta
\partial _{0}\right) \otimes x^{\mu }\partial _{\mu })\times   \notag \\
&&\times \exp \left( \frac{\theta }{2}(x^{\mu }\partial _{\mu }\partial
_{0}\otimes 1+1\otimes x^{\mu }\partial _{\mu }\partial _{0})\right) 
\end{eqnarray}%
For the $\star $-product we can expand all three exponents in this symmetric
twist to be able to check the cyclic property order by order : 
\begin{eqnarray}
f\star _{J_{rs}}g &=&f\cdot g+\frac{\theta }{2}{x^{\mu }}\left( {\partial
_{0}f\cdot \partial _{\mu }g-\partial _{\mu }f\cdot \partial _{0}g}\right) +
\\
&&+\left( \frac{\theta }{2}\right) ^{2}\frac{1}{2}\{{x^{\mu }x^{\rho }}(%
\partial _{\rho }\partial _{\mu }\partial _{0}^{2}f\cdot g+f\cdot \partial
_{\rho }\partial _{\mu }\partial _{0}^{2}g{{+2{{{\partial _{\rho }}}\partial
_{\mu }}f\cdot \partial _{0}^{2}g+{2\partial _{0}^{2}f\cdot {{\partial
_{\rho }}}\partial _{\mu }g+}}}  \notag \\
&&{{{+2\partial _{0}f\cdot {{\partial _{\rho }}}\partial _{\mu }{{\partial
_{0}g}}+2{{\partial _{\rho }}}\partial _{\mu }{{\partial _{0}}}}f\cdot
\partial _{0}g{+{2\partial _{\rho }\partial _{0}f\cdot {{\partial _{0}}}%
\partial _{\mu }g}}+2\partial _{\rho }\partial _{0}^{2}f\cdot \partial _{\mu
}g}+2\partial _{\rho }f\cdot \partial _{\mu }\partial _{0}^{2}g{)}}+ \notag\\
&&+2{x^{\mu }[{\partial _{\rho }\partial _{0}f\cdot \partial _{0}g+{{%
\partial _{0}f\cdot {{\partial _{\rho }}\partial _{0}g}}}}}+3\partial
_{0}^{2}f{\cdot }\partial _{\rho }g{{+\partial _{\rho }f{\cdot \partial
_{0}^{2}g}}}+\partial _{\rho }\partial _{0}^{2}f{\cdot }g+f{\cdot }\partial
_{\rho }\partial _{0}^{2}g]\}+O\left( \theta ^{3}\right)   \notag
\end{eqnarray}

One can notice that the first order is the usual $\kappa $-Minkowski $\star -
$product (see e.g. \cite{Moller}, \cite{AgoGAC}, \cite{Dimitrijevic:2003pn}%
), which coincides for all the symmetric twists (Abelian and Jordanian)
related to the $\kappa $-Minkowski algebra. Therefore we know that one has
to modify the integral introducing the additional measure function $h\left(
x\right) $ such that : 
\begin{equation}
\int h(x)d^{d+1}x\,f\star _{Jrs}g=\int h(x)d^{d+1}x\,g\star _{Jrs}f.
\end{equation}%
At the first order we get:%
\begin{equation}
\left[ \int h(x)(f\star _{J_{rs}}g-g\star _{J_{rs}}f)\right] _{\theta
}=\theta \int h(x){x^{k}}\left( {\partial _{0}f\cdot \partial _{k}g-\partial
_{k}f\cdot \partial _{0}g}\right) 
\end{equation}%
that is, the same expression as for the non-symmetric twist, Eq. (\ref%
{firstord}). Therefore, after integrating by parts we obtain the same
equations for the measure function $h$, (\ref{measure_eq}), as in the
previous case, which again can be shown to be incompatible with the second
order conditions except for the trivial solution $h=0$. 
Therefore we conclude that we do not find any measure of the form $d\mu
=h(x)|\Omega |$ with respect to which $\kappa $-Minkowski star products
deriving from Jordanian twists can be made cyclic.

\section{Conlusions and outlook}

In this paper we have derived a new $\star$-product realizing the $\kappa$%
-Minkowski Lie algebra at the level of space-time coordinate functions,
which is obtained as a symplectic realization in terms of quadratic linear
functions on $\mathbb{C}^d$ with Wick-Voros noncommutativity. We have found
that it corresponds to the $\star$-product provided by a non-symmetric
Jordanian twist constructed in terms of the generators of the two
dimensional Borel algebra $\mathfrak{b}=\mathrm{span}\{P_{0},J;\left[
J\,,P_{0}\right] =P_{0}\}$. The new star product is not cyclic with respect
to the natural integration volume on $\mathbb{R}^d$, $|\Omega|= dx^0
dx^1....dx^d$ nor with respect to $h(x)|\Omega|$, with $h(x)$ any measure
function. Moreover, the lack of cyclicity does not depend on the details of
the Jordanian twist, because we have shown that the same negative result
holds for symmetrized Jordanian twist.

The existence of a star product which is not only cyclic, but also closed
with respect to the corresponding trace functional, which is known as the
generalized Connes-Flato-Sternheimer conjecture, is proven in \cite{FelderSh}%
.\footnote{%
For convenience we report the result
\par
\noindent\textbf{Theorem \cite{FelderSh}} Let $\mathcal{M}$ be a Poisson
manifold with the bivector field $\omega\left( x\right)$, and let $\mathbf{%
\Omega}$ be any volume form on $\mathcal{M}$ such that $\mathrm{div}_\Omega
\omega=0$. Then there exists a star product on $C^\infty(\mathcal{M})$ such
that for any two functions $f$ and $g$ with compact support one has: 
\begin{equation}  \label{FS}
\int (f\star g)\cdot \mathbf{\Omega}= \int f\cdot g\cdot\mathbf{\Omega}.
\end{equation}%
} The problem is how to find it. A possible strategy, see e.g. \cite{Kup15,
KV15}, which we plan to investigate elsewhere, is to start with some
appropriate star product, for example the one we have derived in the present
paper, and then use the gauge freedom \cite{Konts} in the definition of the
star product to obtain the desirable one. Indeed if $\star$ and $%
\star^{\prime }$ are two different star products corresponding to the same
Poisson bi-vector $\omega^{\mu\nu} (x) \partial_\mu\wedge\partial_\nu$ (in
the present case we have $\omega=x^{k}\partial _{k}\wedge \partial _{0}$)
they are related by a local transformation 
\begin{equation}
T\left( f\star g\right)=\left( Tf\star^{\prime }Tg\right) ,  \label{gauge}
\end{equation}
where $T=1+O(\theta)$ is what we shall call the gauge operator. An instance
of such a procedure can be found in \cite{Dito}, where a gauge operator $T$
was constructed, realizing the equivalence between the Gutt star product 
\cite{Gutt} and the Kontsevich one on the dual of Lie algebras. In \cite%
{KV15} this approach is used to determine the gauge operator connecting the
Weyl ordered star-product of $\mathfrak{su}(2)$ type on $\mathbb{R}^3$ to
the closed one.

\section*{Acknowledgments}

The authors would like to thank S. Meljanac for the helpful comments and for
careful reading the paper.\newline
A.P. acknowledges the funding from the European Union's Horizon 2020
research and innovation programme under the Marie Sk{\l }odowska-Curie grant
agreement No 609402 - 2020 researchers: Train to Move (T2M). Part of this
work was supported by National Science Center project 2014/13/B/ST2/04043.
P.V. acknowledges support by COST (European Cooperation in Science and
Technology) in the framework of COST Action MP1405 QSPACE and support by
Compagnia di San Paolo in the framework of the program STAR 2013.

\end{document}